\documentclass[12pt,reqno]{amsart}
\usepackage[utf8]{inputenc}
\usepackage{amsmath, amssymb}
\usepackage{graphicx}
\usepackage{hyperref}
\usepackage{makecell}

\newtheorem{theorem}{Theorem}[section]

\newtheorem{lemma}[theorem]{Lemma}

\newcommand\confidence{\text{confidence}}
\newcommand\correlation{\text{correlation}}
\newcommand\lift{\text{lift}}
\newcommand\maxlift{\text{maxlift}}
\newcommand\quality{\text{quality}}
\newcommand\conf{\text{conf}}
\newcommand\corr{\text{corr}}

\title[Apriori\_Goal Algorithm]{Apriori\_Goal Algorithm for Building Association Rules in a Classified Database}
\author{Vladimir Billig}
\address{Tver State Technical University}
\email{billig.vladimir@gmail.com}
\date{}

\begin{document}

\maketitle

\begin{abstract}
An efficient Apriori\_Goal algorithm is proposed for constructing association rules in a relational database with predefined classification. The target parameter of the database specifies a finite number of goals $Goal_k$, for each of which the algorithm constructs association rules of the form $X \Rightarrow Goal_k$. The quality of the generated rules is characterized by five criteria: two represent rule frequency, two reflect rule reliability, and the fifth is a weighted sum of these four criteria.

The algorithm initially generates rules with single premises, where the correlation criterion between the premise $X$ and the conclusion $Goal_k$ exceeds a specified threshold. Then, rules with extended premises are built based on the anti-monotonicity of rule frequency criteria and the monotonicity of rule reliability criteria. Newly constructed rules tend to decrease in frequency while increasing in reliability. The article proves several statements that justify the rule construction process.

The algorithm enables the construction of both high-frequency and rare rules with low occurrence frequency but high reliability. It also allows for the generation of negative rules with negative correlation between the premise and conclusion, which can be valuable in practical applications for filtering out undesired goals.

The efficiency of the algorithm is based on two factors: the method of encoding the database and its partitioning into subsets linked to the target parameter. Time complexity estimates for rule construction are provided using a medical database as an example.
\end{abstract}

\section{Introduction}

Association rules are a crucial tool in knowledge extraction from databases. Their advantage lies in the fact that extracted knowledge is represented in implication form:

\begin{equation}
X \Rightarrow Y
\end{equation}

where $X$ is the rule premise, and $Y$ is the rule conclusion. This format is intuitive for experts who utilize extracted knowledge, interpreting the rule as "if $X$, then $Y$" or "$X$ implies $Y$." However, a rule is valuable only if it holds with high probability.

Association rules extracted from databases gained widespread use after the introduction of the Apriori algorithm, proposed by R. Agrawal and colleagues \cite{agrawal1993, agrawal1994, agrawal1995}. The Apriori algorithm considers a finite set of literals:

\begin{equation}
D = \{d_1, d_2, \dots, d_m\}
\end{equation}

and a database:

\begin{equation}
db = \{t_1, t_2, \dots, t_n\}
\end{equation}

consisting of records (transactions), where each transaction $t_k$ is a subset of literals.

In an association rule (1), $X$ and $Y$ are non-empty and non-overlapping subsets of $D$:

\begin{equation}
X \cap Y = \emptyset
\end{equation}

For each rule constructed by the Apriori algorithm, two characteristics determine its value: frequency and confidence. Frequency is defined as the ratio of transactions that contain both the premise and conclusion of the rule to the total number of transactions. Confidence is the ratio of transactions that contain both the premise and conclusion to those that contain only the premise.

The objective of the Apriori algorithm is to generate all association rules whose frequency and confidence exceed predefined minimum thresholds.

The efficiency of the Apriori algorithm, which enables it to manage the combinatorial explosion of possible rules, is based on the property of anti-monotonicity of frequency. If a set $Z$ is not frequent, then any of its supersets (extensions of $Z$) are also not frequent. As a result, the algorithm first identifies all frequent subsets of $D$, starting with single frequent items and expanding them iteratively. In the second stage, these frequent subsets are used to generate rules that satisfy the required confidence level.

Frequency is the primary criterion: if a set of literals is not frequent, it is excluded from rule generation, even if it has a high probability of occurrence.

Subsequent variations of association rule mining algorithms have, in one way or another, built upon the fundamental idea of the Apriori algorithm. These variations have evolved in different directions.

One primary focus has been improving the efficiency of the algorithm to make it applicable to large databases. The key idea behind this optimization is to construct a data structure that allows rule generation using only a subset of the original database records. Several studies have contributed to this direction. For instance, F. Bodon \cite{bodon2003} proposed a sorting tree structure that demonstrated its effectiveness in database compression. Another efficient implementation of the Apriori algorithm was developed by Christian Borgelt \cite{borgelt2002, borgelt2003, borgelt2004}, in which the database is compressed into a transaction tree. A distinctive feature of these works is their optimization of Apriori into a highly efficient implementation written in C. Borgelt’s survey \cite{borgelt2012} provides an overview of various algorithms for frequent pattern mining.

The database model in the classic Apriori algorithm is based on a market basket analysis framework, where each record represents a shopping basket in a supermarket. The literals in this model correspond to individual products, each assigned a unique identifier (ID). Given that the number of products can be quite large, a hierarchical structure is often introduced to group products at different levels. Several studies have focused on the construction of multi-level association rules. The first notable work in this area was by J. Han and Y. Fu \cite{han1999}. Later studies, such as \cite{shrivastava2013}, analyzed the efficiency of multi-level rule construction. The seminal book by J. Han, M. Kamber, and J. Pei \cite{han2012}, which covers fundamental concepts and methodologies in data mining, provides an in-depth discussion of various approaches to association rule mining within a broader data mining framework.

However, the itemset-based database model is a special case of a more general relational database structure. It is desirable to apply association rule mining to relational databases as well. Such cases are explored in the work of Minakshi Kaushik and co-authors \cite{kaushik2021}, where the authors address how to adapt Apriori to datasets with different attribute types, including continuous and categorical attributes. Their study examines various approaches to transforming such data into the itemset format required by the Apriori algorithm.

As mentioned earlier, a fundamental characteristic of the Apriori algorithm is the preprocessing of frequent itemsets before rule generation. However, this constraint results in the loss of potentially valuable rules that exhibit low frequency but high confidence—rules that may be particularly insightful for experts. Identifying such rare but highly reliable rules is often of significant practical interest. Studies such as \cite{antonello2021, borah2020, capillar2023} focus on mining rare association rules in databases where item frequencies vary widely.

Another area of research involves the discovery of negative association rules, where the presence of one attribute reduces the likelihood of another. Studies \cite{qiu2018, huang2022} address techniques for mining negative patterns, which can be just as valuable in practice as positive association rules.

Association rule mining has also been extensively applied to solving real-world problems across various domains. In particular, several studies \cite{piran2024, billig2016, billig2017} focus on extracting rules from medical databases, where association patterns can help identify important diagnostic and predictive factors. Such databases serve as a model framework for the Apriori\_Goal algorithm proposed in this paper.

 How is Apriori\_Goal Different?

We do not propose an optimized version of Apriori that extracts the same rules more efficiently. Instead, the Apriori\_Goal algorithm constructs association rules based on different criteria, making them more relevant to practical applications.

The main objective of Apriori\_Goal is to construct rules with increasing confidence. Unlike the classical Apriori algorithm, where frequency is the primary filtering criterion, Apriori\_Goal does not rely on frequency as the main constraint. The rules generated can be both frequent and rare, provided they exhibit strong confidence levels.

Furthermore, the classical Apriori algorithm treats all literals equivalently, allowing the construction of both $X \Rightarrow Y$ and $Y \Rightarrow X$. However, in our model, the database records represent object properties. These properties can be either categorical or continuous, but one key feature is the presence of a target parameter (Goal) that defines the classification. All other properties are considered input parameters, forming the premise of the rule.

Thus, our task is to identify a premise $X$, consisting of input parameters, that strongly predicts a given goal $Goal_k$:

\begin{equation}
X \Rightarrow Goal_k
\end{equation}

Additionally, Apriori\_Goal constructs negative rules, where the presence of $Y$ strongly negates the occurrence of $Goal_k$:

\begin{equation}
Y \Rightarrow \neg Goal_k
\end{equation}

We believe that databases with a target parameter represent the general case characteristic of real-world databases, whereas the absence of a target parameter is a special case that requires separate consideration. In our model, a medical database serves as the primary example, where patients represent the objects, and records describe their attributes, such as Age, Gender, Temperature, Blood Pressure, and Lab Test Results. The target parameter in this context is the Diagnosis, determined by a physician.

In designing the Apriori\_Goal algorithm, we prioritized sufficient efficiency, opting not to construct a complex database structure aimed solely at minimizing rule generation time. Instead, our algorithm's efficiency is achieved through two key factors:

Parallelization: Since rule generation for different values of $Goal_k$ can be executed independently, the algorithm allows parallel processing. This enables each parallel thread to work with only a subset of the database, eliminating the need for any single thread to traverse the entire dataset.

Optimized Encoding of Binary Attributes: The encoding scheme used in Apriori\_Goal compresses the original database by representing each record as a single integer, regardless of its length, while still preserving all necessary information about the binary attributes of the objects in the dataset. This same encoding method is applied to attribute sets forming the premise of a rule.

This encoding technique significantly enhances computational efficiency by replacing set-based operations with bitwise logical operations on integers, which are executed almost instantaneously.

This paper provides a detailed description of the Apriori\_Goal algorithm, building upon the ideas presented in previous works \cite{billig2016, billig2017}.

The structure of this paper is as follows:

Section 2 discusses the characteristics of databases with a target parameter, along with the preprocessing tasks performed at the initial stage of the algorithm.
Section 3 introduces the evaluation criteria used to assess the quality of the generated association rules.
Section 4 explores the data structures and core principles behind the algorithm’s construction.
Section 5 examines the efficiency factors, presenting experimental performance evaluations obtained from a real-world database.
Section 6 provides theoretical justifications supporting the correctness of the proposed rule construction method.
Finally, the Conclusion summarizes the key findings.

\section{Database and Preprocessing}

\subsection{Database Encoding}

The Apriori\_Goal algorithm is designed for a general database schema, where a database is viewed as a collection of objects described by their properties. The set \( D \) represents the set of all properties:

\begin{equation}
D = \{P_1, P_2, \dots, P_m\}
\end{equation}

The database itself can be represented as a typical relational table with \( N \) rows (records) and \( M \) columns (properties of objects). During the preprocessing stage, each property is transformed into a set of binary attributes, which can be interpreted as literals in the classical Apriori algorithm. The resulting set of binary properties is considered an enumeration, where each element is assigned a unique integer code.

Each database record, represented by a set of binary properties, is encoded as a single integer, computed as the sum of the assigned codes of all included properties. The details of this encoding process are discussed in Section 2.3.

\subsection{Databases with a Target Parameter}

The Apriori\_Goal algorithm is specifically designed for databases that include a target parameter. This target parameter is a categorical attribute with a finite set of values, each representing a distinct goal. The presence of a target parameter allows the classification of database records. By sorting the records according to their target parameter values, the database can be partitioned into subsets:

\begin{equation}
db = \{db_1, db_2, \dots, db_p\}
\end{equation}

Each subset \( db_k \) consists of records that share the same target parameter value.

For example, in a medical database, the target parameter is the diagnosis, and association rules can help identify significant patterns among patient attributes that correlate with specific diagnoses. Similarly, in a banking database containing customer information, association rules can classify clients and recommend appropriate financial products.

\subsection{Preprocessor}

The preprocessor plays a crucial role in data preparation. Its functions include:

\begin{itemize}
    \item Reading the raw database, which may be stored in different formats such as Excel, CSV, or other structured files.
    \item Sorting the records based on the values of the target parameter.
    \item Converting continuous and categorical attributes into binary attributes.
    \item Encoding binary attributes into numerical values.
\end{itemize}

\subsubsection{Transformation of Attributes into Binary Form}

Each column in the database can be of either categorical or continuous type. Continuous attributes are first discretized into categorical values. This discretization process involves dividing the range of continuous values into \( k \) intervals (bins). 

For example, an attribute representing body temperature might be classified into three categories: 
- Below normal, 
- Normal, and 
- Above normal.

Each transformed categorical attribute is then mapped to \( k \) separate binary attributes. The names of these binary attributes are automatically generated based on the original attribute name and the category index.

For example, for the Temperature attribute (with three categories), the database will contain three new binary columns:
- \( T_0 \) (temperature below normal),
- \( T_1 \) (normal temperature),
- \( T_2 \) (temperature above normal).

For any given record, only one of these binary attributes will have a value of 1, while the others will be set to 0.

The resulting set of binary properties is:

\begin{equation}
P = \{P_1, P_2, \dots, P_m\}
\end{equation}

This set can be treated as a collection of literals in the classical Apriori algorithm.

\subsubsection{Encoding of Binary Attributes}

Binary attributes are numerically encoded using the following scheme: each binary attribute indexed by \( k \) in the enumeration is assigned a code equal to \( 2^k \).

Thus, a database record can be represented as a single integer \( N \), computed as:

\begin{equation}
N = \sum_{j=1}^{m} 2^{(j-1)}
\end{equation}

where summation is performed over all \( j \) indices where \( P_j = 1 \).

This encoding ensures that \( N \) is always within the range:

\begin{equation}
0 < N < 2^m
\end{equation}

In the binary representation of \( N \), the positions of ones correspond to the active binary attributes of the record.

This encoding technique is applied to both database records and itemsets forming the premise of association rules.

\subsubsection{Practical Considerations}

A potential limitation of this encoding method is that the number \( N \) can become very large. For example, if \( m = 100 \) binary attributes are used, the value of \( N \) could require over 30 digits, exceeding the standard integer limits in many programming languages.

However, this issue is easily addressed in modern programming environments:
- In Python, integers have no fixed size limit, making large numbers manageable.
- In C\#, the BigInteger type efficiently handles large numbers and is used in our implementation.

\subsubsection{Additional Information for Preprocessing}

To perform all required transformations, the preprocessor requires additional metadata. In our implementation, this metadata is provided via two input files:

\begin{itemize}
    \item Database File (db) – Contains the raw dataset.
    \item Database Description File (dbd) – Provides metadata about the database structure.
\end{itemize}

The dbd file contains one row per column of the original database, with each row specifying:

\begin{itemize}
    \item Column Type: Can be target, continuous, or categorical.
    \item Short Name: Used for naming attributes in the association rules.
    \item Number of Classes: For categorical and target columns, this specifies the number of possible values. For continuous columns, it defines the number of discretization bins.
    \item Value List: 
        - For the target column, this lists the size of each database subset.
        - For categorical attributes, this enumerates all possible values.
        - For continuous attributes, this defines the interval boundaries used for discretization.
    \item Full Name: A human-readable description of the attribute, useful for interpreting generated rules.
\end{itemize}

The information provided in the dbd file is sufficient to convert every database record into an integer representation \( N \), encoding its binary properties. 

Moreover, the conversion process supports parallelization, allowing multiple records to be processed simultaneously.

\section{Association Rules and Their Characteristics}

\subsection{Association Rules}

In the Apriori\_Goal algorithm, association rules take the form:

\begin{equation}
X \Rightarrow Goal_k
\end{equation}

The target of the rule is fixed, and the algorithm searches for the premise \( X \), which is a set of input parameters that are associated with the given target \( Goal_k \).

There are many possible variations of the set \( X \). The key challenge is identifying "good" sets that generate meaningful rules and reveal hidden knowledge within the original database. The quality of a rule is determined by specific criteria that can be computed for a given association rule.

\subsection{Rule Quality Criteria}

In the classical Apriori algorithm, rule quality is characterized by two key metrics computed from the database: frequency and confidence.

In the Apriori\_Goal algorithm, rule quality is assessed using five criteria. Two of these criteria, \( f_g \) and \( f_{all} \), measure the frequency of the rule. The first criterion, \( f_g \), represents the frequency of occurrence of the combined set \( X \cup Goal_k \) within the subset \( db_k \), which corresponds to the target \( Goal_k \). The second criterion, \( f_{all} \), expresses the frequency of this set relative to the entire database.

The next two criteria, confidence and correlation, assess the reliability of the rule from different perspectives.

A fifth, generalized criterion, quality, is constructed as a weighted sum of these four criteria.

The fundamental operation in the Apriori\_Goal algorithm is the computation of the support of the set \( X \) in either the entire database or a subset of it. Support is defined as the number of records containing the set \( X \):

\begin{equation}
Support(X, db) = |\{Z_k \mid X \subseteq Z_k, Z_k \in db\}|
\end{equation}

In this definition, \( Z_k \) represents the records of the analyzed database. Conceptually, this operation determines whether a given record \( Z_k \), which consists of a set of binary properties, contains the subset \( X \) of binary properties.

All rule quality metrics are derived using this fundamental support operation.

\begin{equation}
f_g = \frac{Support(X, db_k)}{n_k}, \quad f_{all} = \frac{Support(X, db_k)}{N} = f_g \cdot \frac{n_k}{N}
\end{equation}

The frequency metrics are constrained by the following ranges:

\begin{equation}
0 \leq f_g \leq 1, \quad 0 \leq f_{all} \leq \frac{n_k}{N}
\end{equation}

It is important to note that both frequency criteria rely on the same fundamental operation, as the combined set \( X \cup Goal_k \) appears only within the records of the subset \( db_k \), which corresponds to the given target.

Both criteria are valuable for experts analyzing association rules. In medical applications, for example, \( f_g \) represents the frequency of occurrence of property set \( X \) among patients diagnosed with \( Goal_k \), while \( f_{all} \) measures how often \( X \) appears across all patients.

A third key metric for association rule evaluation is confidence, which measures the reliability of the rule. It is defined as the proportion of records containing \( X \) within the subset \( db_k \) compared to all records in the database that contain \( X \):

\begin{equation}
\text{confidence}(X \Rightarrow Goal_k) = \frac{Support(X, db_k)}{Support(X, db)}
\end{equation}

The confidence metric falls within the range:

\begin{equation}
0 \leq \text{confidence} \leq 1
\end{equation}

The upper bound is reached when \( X \) occurs exclusively in records belonging to subset \( db_k \). Even when the frequency of occurrence is low, rules with high confidence may be of great interest to experts, as they indicate that the set \( X \) strongly correlates with the target \( Goal_k \).

Another important criterion, known as lift, is based on the probabilistic interpretation of attributes. The set of attributes \( X \) is treated as a multidimensional random variable that appears with a certain probability. The lift criterion is defined as the ratio of the conditional probability of the target given the premise \( X \) to the overall probability of the target:

\begin{equation}
\lift = \frac{P(Goal_k \mid X)}{P(Goal_k)}
\end{equation}

Switching from probabilities to frequencies, we obtain:

\begin{equation}
P(Goal_k \mid X) = \frac{\text{Support}(X, db_k)}{\text{Support}(X, db)}
\end{equation}

\begin{equation}
P(Goal_k) = \frac{n_k}{N}
\end{equation}

It follows that the lift criterion can be expressed as:

\begin{equation}
\text{lift} = \text{confidence} \times \frac{N}{n_k}
\end{equation}

It is interesting to note that the two frequency-related criteria and the two reliability-related criteria are linked by a single linear dependence.

Despite this linear dependence between the paired criteria, each criterion provides valuable information to experts due to its distinct interpretational significance.

The lift criterion can be interpreted as a measure of correlation between the premise and conclusion of the rule. Its values fall within the range:

\begin{equation}
0 \leq \text{lift} \leq \frac{N}{n_k}
\end{equation}

There are three critical points for the lift criterion:

$\bullet$ lift $= 1$. This occurs when the conditional probability \( P(Goal_k \mid X) \) equals the marginal probability \( P(Goal_k) \), meaning that the events are independent and there is no correlation between the premise and the conclusion. The correlation coefficient in this case is zero.
  
$\bullet$ lift $= 0$.  This means that the probability of the target \( Goal_k \) is zero when the premise \( X \) appears. In this case, the correlation is negative. From an expert's perspective, this suggests that the occurrence of \( X \) negates the possibility of \( Goal_k \). The correlation coefficient is \(-1\).
  
$\bullet$ lift $= 1 / P(Goal_k)$. The maximum value of the lift criterion occurs when the conditional probability of the target given \( X \) equals 1, meaning that whenever \( X \) appears, \( Goal_k \) always follows. In this case, the correlation between the premise and conclusion is positive with a coefficient of 1.

For experts, correlation between the premise and the conclusion of a rule is often easier to interpret than conditional probabilities. Therefore, the correlation criterion is introduced as the fourth characteristic of association rules. This criterion is computed as a function of lift. Using the critical points, we define a piecewise linear function as follows:

\begin{equation}
\text{correlation} = \text{lift} - 1, \quad \text{if } 0 \leq \text{lift} \leq 1
\end{equation}

\begin{equation}
\text{correlation} = \frac{\text{lift} - 1}{\text{maxlift} - 1}, \quad \text{if } 1 \leq \text{lift} \leq \text{maxlift}
\end{equation}

The values of the correlation criterion fall within the range:

\begin{equation}
-1 \leq \correlation \leq 1
\end{equation}

The value of \( \maxlift \), which equals \( 1 / P(Goal_k) \), is computed as:

\begin{equation}
\maxlift = \frac{N}{n_k}
\end{equation}

A fifth characteristic, quality, is introduced as a generalized rule evaluation criterion. It is computed as a weighted sum of the four previously defined criteria:

\begin{equation}
\quality = p_1 \times f_{all} + p_2 \times f_g + p_3 \times \confidence + p_4 \times \correlation
\end{equation}

The weights for these criteria are set by experts and, by default, are all equal to 1.

The introduction of the correlation criterion allows for the classification of association rules based on their target \( Goal_k \). The input parameters and their combinations forming the premise of a rule can be categorized into three broad classes:

$\bullet$ Null-correlated premises, which do not provide meaningful insights and are generally of little interest to experts.

$\bullet$  Positively correlated premises, which exhibit a strong association between the premise and the conclusion and are typically of high interest.

$\bullet$  Negatively correlated premises, which may also be valuable, as they help exclude a target conclusion when a specific premise is observed.

\section{Data Structures and Fundamental Ideas of the Algorithm}

\subsection{Preprocessing Stage}

At this stage, the preprocessor takes as input two files: one containing the raw database and another containing its description. The output of the preprocessor is an encoded database, represented as a list of integers. The database records are sorted and divided into \( k \) subsets according to the number of target parameter values. The preprocessor also constructs an enumeration of binary properties and generates an array of \( k \) elements that define the sizes of the database subsets.

\subsection{Main Processing Stage}

At this stage, the Rule class and the Apriori module are constructed.

\subsubsection{The Rule Class}

Objects of this class represent association rules. The class fields that define the properties of its objects include:

\begin{itemize}
    \item The rule premise \( X \), which is encoded as an integer, just like a database record, preserving information about all binary properties that make up the premise.
    \item The target index \( k \).
    \item Five criteria that define the characteristics of the rule.
\end{itemize}

The class constructor, given a premise \( X \), a target index \( k \), and the support values \( sup_k \) (support of \( X \) in the target subset) and \( sup \) (support of \( X \) in the entire database), constructs the corresponding rule and computes all its characteristic values.

\subsubsection{The Apriori Module}

The Apriori module defines methods for constructing rules and the data structures processed by these methods.

The core method responsible for rule construction follows the scheme:

\begin{verbatim}
void CreateRules()
{
    CreateCandidates();       
    bool existNewRules = true;
    while (existNewRules)
    {
        NextCurrent();
        for (int k = 0; k < P; k++)
        {
            current[k].Clear();
            for (int i = 0; i < nextcurrent[k].Count; i++)
            {
                rules[k].Add(nextcurrent[k][i]);
                current[k].Add(nextcurrent[k][i]);
            }
        } 
        existNewRules = Check();
     } 
}
\end{verbatim}

The fundamental idea behind the rule construction algorithm is simple and intuitive. Initially, a set of candidates Candidates is generated. This set serves as the initial value of the Current set, allowing the algorithm to generate rules with a premise of length 1. Then, within a loop, the premise length is incrementally extended by one.

At each iteration, the NextCurrent set is generated using the Current set (which contains premises of length \( k \)) and the Candidates set. The newly constructed rules are added to the Rules structure. The NextCurrent set is then transferred to Current, which is cleared to be populated in the next iteration. 

The loop always terminates, as the maximum length of the premise is bounded and cannot exceed the number of candidates. In practice, the loop terminates earlier as soon as no new rules are generated in an iteration.

In the AprioriGoal algorithm, candidates are binary properties that form positively correlated rules, meaning that the correlation between the rule premise and the target exceeds a specified threshold.

Extending a positively correlated rule with another positively correlated candidate results in a new rule with an increased correlation value. The correlation criterion exhibits monotonicity, meaning that the correlation can only increase when rules are extended. The theoretical justification for this property will be provided in the next section.

Due to this monotonicity, the correlation can rapidly reach its maximum value of one. In this case, the rule is considered final and is no longer expanded. The algorithm allows specifying a correlation threshold, such that once this threshold is reached, further expansion stops.

Conversely, due to the anti-monotonicity of frequency, extending a rule may reduce its frequency, potentially reaching the lower limit of zero. The algorithm allows specifying a minimum frequency threshold, below which the rule is considered final and is no longer expanded.

The methods CreateRules, CreateCandidates, and NextCurrent operate on four fundamental data structures, each represented as an array of lists where the elements are objects of the Rule class:

\begin{itemize}
    \item \texttt{List<Rule>[] candidates} – The set of candidates for rule premises. This includes rules with a single premise that have passed the threshold set by the correlation criterion.
    \item \texttt{List<Rule>[] current} – The current set of rules with a premise of length \( k \). When \( k = 1 \), this set is identical to the candidates set.
    \item \texttt{List<Rule>[] nextcurrent} – The expanded set of rules with a premise of length \( k+1 \), generated from the current and candidates sets.
    \item \texttt{List<Rule>[] rules} – The complete set of all generated rules.
\end{itemize}

In all structures, the list corresponding to the \( k \)-th element of the array contains rules associated with the \( k \)-th target.

\section{Efficiency}

The goal of the Apriori\_Goal algorithm is to provide experts with a convenient tool for discovering and analyzing association rules with high efficiency.

One of the key factors ensuring the efficiency of the algorithm is the method used for encoding binary properties. Database records and rule premises, which represent sets of binary properties, are encoded as integers. This encoding approach replaces computationally expensive set operations with bitwise operations on integers, which are executed almost instantaneously.

Another important factor influencing efficiency is the presence of a target parameter, which allows the database to be partitioned into subsets. To compute the characteristics of generated rules, it is sufficient to process these database subsets in parallel, each associated with a specific target.

\subsection{Efficiency of the Fundamental Operation}

The efficiency of the algorithm is largely determined by the performance of its fundamental operation, which computes the support of a given set and requires traversal of the database.

The Support method, which calculates the support of a premise \( X \) in the database, combines both efficiency factors, significantly enhancing algorithm performance.

Below is the implementation of this method from the C\# implementation of the algorithm:

\begin{verbatim}
public static (int[], int) Support(BigInteger X)
{
    int[] sup_k = new int[P];
    int sup = 0;
    Parallel.For(0, P, (k) =>
    {
        int start = startk[k];
        int finish = start + nk[k];
        for (int i = start; i < finish; i++)
            if ((X & db[i]) == X)
                sup_k[k] += 1;
    });
    for (int i = 0; i < P; i++)
        sup += sup_k[i];
    return (sup_k, sup);
}
\end{verbatim}

In this method, \( P \) parallel threads are executed, each processing its corresponding subset of the database. Each thread computes \( sup_k[k] \), which represents the support of premise \( X \) within a specific subset. Once all threads complete, the overall support \( sup \) is computed by summing the elements of the array.

A crucial aspect of this method is the efficient verification of whether a database record \( db[i] \) contains the property set \( X \). This verification is performed using a logical condition on integers, avoiding explicit set operations within the algorithm.

\subsection{Efficiency of Rule Premise Expansion}

Another key factor influencing the overall efficiency of the algorithm is the process of expanding rule premises.

A rule premise of length \( k \) is defined as a set of \( k \) binary properties. These binary properties are ordered according to the predefined sequence of binary properties. A premise can only be expanded by candidates whose index in the enumeration is greater than the highest-indexed binary property currently in the premise.

The encoding method used for binary properties allows premise expansion without analyzing the properties already included in the premise. The candidate selection process is simplified to a comparison between the candidate value and the premise value. The candidate's encoded integer value must be greater than the encoded value of the premise, which represents the sum of the codes of all binary properties in the premise.

Below is a fragment of the program code that determines the valid candidates for premise expansion and executes the expansion process:

\begin{verbatim}
num = 0;
while (num < candidates[k].Count && 
    rule.X >= candidates[k][num].X) num++;
for (int q = num; q < candidates[k].Count; q++)
{
    Z = rule.X + candidates[k][q].X;
    ...
} 
\end{verbatim}

The while loop determines the minimum index of a candidate that is eligible for expanding the premise. The for loop then expands the premise using valid candidates. The new premise \( Z \) of length \( k+1 \) is constructed by simply adding two integers. This approach ensures maximum efficiency, as the expansion operation consists of just one integer addition.

\subsection{Complexity Analysis}

The Support method, which computes the support of a rule premise, has a complexity of:

\begin{equation}
O(N / k)
\end{equation}

where \( N \) is the number of records in the database.

The rule constructor, which computes the characteristics of a rule, performs a finite number of operations and has a complexity of:

\begin{equation}
O(1)
\end{equation}

Thus, the overall runtime complexity of the algorithm that generates \( M \) rules is:

\begin{equation}
O(M \cdot N / k)
\end{equation}

\subsection{Experimental Evaluation of Rule Construction Time}

We present the results of an experimental evaluation of rule construction time using the {\it diabetes} dataset, available in the Scikit-learn package in Python \cite{sklearn}. This dataset includes a target parameter and 10 input parameters of both continuous and categorical types, containing a total of 442 records. By duplicating the records, we expanded the dataset to 4 million records.

Table \ref{tab:runtime} shows the rule search time as a function of the database size and the number of rules generated.

\begin{table}[h]
\centering
\begin{tabular}{|c|c|c|}
\hline
\textbf{Number of Records} & 
\makecell{\textbf{Number of}  \\ \textbf{Rules Generated}}
 &  \makecell{\textbf{Rule Construction} \\ \textbf{Time (sec)}} \\
\hline
44200 & 13 & 0.015 \\
44200 & 22 & 0.025 \\
442000 & 13 & 0.17 \\
442000 & 22 & 0.20 \\
4420000 & 13 & 1.85 \\
4420000 & 22 & 2.04 \\
\hline
\end{tabular}
\caption{Rule construction time}
\label{tab:runtime}
\end{table}

A total of 22 rules were generated from a 4-million-record database in approximately 2 seconds. The computations were performed on a standard Dell laptop.

This level of performance was achieved using the full version of the algorithm, which operates with the BigInteger type. Given that the number of input parameters in this dataset is relatively small, a lightweight version using the ulong (unsigned long integer) type in C\# could be employed, which would further speed up computations.

\section{Correctness}

We now prove statements that hold for the Apriori\_Goal algorithm.

\begin{theorem}[Anti-Monotonicity of Rule Frequency]
When a positively correlated rule with premise \( X \) is expanded using a positively correlated candidate \( Y \), the frequency of the expanded rule \( f_g(X, Y) \) can only decrease relative to the frequencies of the original rules and satisfies the following bounds:
\begin{equation}
\max(0, f_g(X) + f_g(Y) -1 ) \leq f_g(X, Y) \leq \min(f_g(X), f_g(Y))
\end{equation}
\end{theorem}
\begin{proof}
The frequency of the expanded rule is proportional to the number of records in the subset \( db_k \) that contain both \( X \) and \( Y \). The maximum possible number of such records is achieved when the overlap between the records containing \( X \) and those containing \( Y \) is maximal. This observation establishes the validity of the upper bound in inequality (31).

The lower bound in (31) holds in the case of minimal overlap. In the extreme case, this overlap may be empty, but when the frequencies of the original rules are high, there will still be records that contain both \( X \) and \( Y \).

Thus, inequality (31) confirms that rule frequency does not increase when the premise is expanded; in the best case, it remains at the frequency of the less frequent of the two original rules.
\end{proof}
We now establish several statements that describe the behavior of rule reliability criteria as the rule premise is expanded. Without loss of generality, we assume that the database is partitioned into two subsets: one associated with the target \( Goal_k \) and the other associated with its negation \( \neg Goal_k \).

We begin with the following simple lemma:

\begin{lemma}
If the confidence of the rule \( X \Rightarrow Goal_k \) is greater than 0.5, then the probability of \( X \) appearing in the subset \( db_k \) is higher than its probability in the opposite subset.
\end{lemma}
\begin{proof}
By definition, confidence is given by:

\begin{equation}
\conf = \frac{a}{a + b}
\end{equation}

where:
- \( a \) is the number of occurrences of \( X \) in the subset \( db_k \),
- \( b \) is the number of occurrences of \( X \) in the opposite subset.

If:

\begin{equation}
\frac{a}{a + b} > 0.5
\end{equation}

then it follows that \( a > b \), which proves the lemma.
\end{proof}

Now, let \( n_k \) and \( n \) represent the sizes of the two subsets, and define \( p = \frac{n}{n_k} \) as their ratio. We then establish the following result:

\begin{lemma}
If the correlation criterion satisfies:
\begin{equation}
\text{\rm corr} > \frac{p-1}{2p}
\end{equation}
then the confidence criterion satisfies:
\begin{equation}
\text{\rm conf} > 0.5
\end{equation}
\end{lemma}
\begin{proof}
From the definition of lift, we have:

\begin{equation}
\lift = \conf \times \frac{n_k + n}{n_k} = \conf \times (1 + p)
\end{equation}

For a positively correlated rule, we know that \( \lift > 1 \). Indeed, if \( \lift \leq 1 \), then by definition:

\begin{equation}
\corr = \lift - 1 \leq 0
\end{equation}

which contradicts the assumption of positive correlation.

Since \( \lift > 1 \), we use the definition:

\begin{equation}
\maxlift = \frac{N}{n_k} = 1 + p
\end{equation}

\begin{equation}
\corr = \frac{\lift - 1}{\maxlift - 1} = \frac{\conf \times (1 + p) - 1}{p}
\end{equation}

Since correlation is positive, we derive the bound:

\begin{equation}
\conf > \frac{1}{1 + p}
\end{equation}

If \( p < 1 \), then confidence is always greater than 0.5.

If \( p > 1 \), then positive correlation alone is not sufficient to guarantee that confidence exceeds 0.5. Let \( d \) be the correlation threshold. What threshold \( d \) must be set for a given \( p \) to ensure that confidence exceeds 0.5?

It follows that the threshold must satisfy:

\begin{equation}
d > \frac{p-1}{2p}
\end{equation}

which proves the statement.
\end{proof}

For \( p = 2 \), the correlation criterion must exceed 0.25. If the correlation criterion exceeds 0.5, then confidence is guaranteed to be above 0.5. According to Lemma 6.2, in this case, the probability of \( X \) appearing in the subset \( db_k \) is greater than its probability in the opposite subset.

Let \( X \) and \( Y \) be the premises of two rules treated as independent multidimensional random variables. We will call them \( Goal_k \)-independent if the following equalities for conditional probabilities hold:
\begin{align*}
P(X \& Y / Goal_k) &= P(X / Goal_k) P(Y / Goal_k), \\
P(X \& Y / \neg Goal_k) &= P(X / \neg Goal_k) P(Y / \neg Goal_k).
\end{align*}

 Then the following holds:

\begin{theorem}[Monotonicity of Confidence and Correlation Criteria]
Let $X$ and $Y$ be $Goal_k$-independent  premises of two positively correlated rules. When a rule with premise \( X \) is expanded using \( Y \), the confidence of the expanded rule can only increase relative to the confidence of the original rules.
\end{theorem}
\begin{proof}
Using definition (16), confidence can be expressed as:

\begin{equation}
\conf(X) = \frac{P(X) \cdot n}{P(X) \cdot n + Q(X) \cdot (N-n)} = \frac{a}{a+b}
\end{equation}
where:
- \( P(X) \) is the probability of \( X \) appearing in records within the subset \( db_k \) associated with the target \( Goal_k \),
- \( n \) is the size of this subset,
- \( Q(X) \) is the probability of \( X \) appearing in database records not associated with \( Goal_k \),
- \( N \) is the total number of database records.

Now, we compute the difference:

\begin{align}
&\conf(X, Y) - \conf(X) = \nonumber \\
& \frac{P(X) P(Y) n}{P(X) P(Y) n + Q(X) Q(Y) (N - n)}
- \frac{P(X) n}{P(X) n + Q(X) (N - n)}
\end{align}

Simplifying, we obtain:

\begin{equation}
\conf(X, Y) - \conf(X) = c_1 (P(Y) - Q(Y))
\end{equation}

Similarly,

\begin{equation}
\conf(X, Y) - \conf(Y) = c_2 (P(X) - Q(X))
\end{equation}

The confidence of the expanded rule exceeds the confidence of the original rules if the probabilities of \( X \) and \( Y \) occurring in the subset \( db_k \) are greater than their probabilities in records associated with other targets.

Lemma 6.3 confirms that these conditions hold when a lower correlation threshold is correctly set for candidate rules with single premises. Consequently, the probability differences in (44) and (45) are positive. An increase in confidence leads to an increase in correlation, proving the theorem.
\end{proof}

Since confidence is monotonic and frequency is anti-monotonic, rule construction terminates quickly, as confidence either reaches the specified upper threshold or frequency reaches the specified lower threshold.

\subsection{Independence Assumption and Practical Considerations}

In the proof of Theorem 6.4, the assumption that premises \( X \) and \( Y \) are independent was essential. If the premises are dependent, then instead of full probabilities, conditional probabilities appear in equations (44) and (45), for which Lemma 6.3 may no longer hold.

In real databases, some parameters may be correlated, but their number is usually small. In an expanded premise \( X \), which includes both correlated and uncorrelated parameters with candidate \( Y \), the overall correlation effect remains minor.

\subsection{Experimental Illustration}

As an example, we present rule construction results for the {\it diabetes} dataset mentioned earlier. Rules were constructed using a minimum correlation threshold of 0.35.

\begin{table}[h]
\centering
\begin{tabular}{|c|c|c|c|c|c|}
\hline
Rule &  f\_g & f\_all & conf & corr & q \\ 
\hline
BMI0  $\Rightarrow$ Goal0  & 0.757 & 0.353  & 0.658 & 0.36 & 2.128 \\
S50   $\Rightarrow$ Goal0  &  0.335  & 0.156  & 0.758 & 0.547 & 1.797 \\
BMI0, S50  $\Rightarrow$ Goal0  & 0.291 & 0.136  & 0.833 & 0.688 & 1.948 \\
BMI2  $\Rightarrow$ Goal2  & 0.25 & 0.045   & 0.769 & 0.718 & 1.783 \\
BP2   $\Rightarrow$ Goal2  & 0.488 & 0.088   & 0.488 & 0.374 & 1.437 \\
S42  $\Rightarrow$ Goal2   & 0.138 & 0.025  & 0.579 & 0.486 & 1.227 \\
S62  $\Rightarrow$ Goal2   & 0.425 & 0.077   & 0.472 & 0.356 & 1.33 \\
BMI2, BP2  $\Rightarrow$ Goal2   & 0.175 & 0.032  & 1 & 1 & 2.207 \\
BMI2, S42  $\Rightarrow$ Goal2   & 0.038 & 0.007  & 1 & 1 & 2.044 \\
BMI2, S62  $\Rightarrow$ Goal2   & 0.1 & 0.018   & 0.8 & 0.756 & 1.674 \\
BP2, S42  $\Rightarrow$ Goal2  & 0.05 & 0.009  & 0.571 & 0.477 & 1.107 \\
BP2, S62  $\Rightarrow$ Goal2   & 0.2 & 0.036   & 0.64 & 0.56 & 1.437 \\
S42, S62  $\Rightarrow$ Goal2  & 0.088 & 0.016  & 0.7 & 0.634 & 1.437 \\
\hline
\end{tabular}
\caption{Association rules built on {\it diabetes} database}
\label{tab:AR}
\end{table}

For Goal0, two candidates surpassed the threshold with correlation values of 0.36 and 0.55. Expanding the rule by combining these candidates increased the correlation to 0.69, while reducing frequency \( f_{all} \) by approximately a factor of three. At this point, rule construction for Goal0 stops, as there are no remaining candidates that exceed the threshold.

For Goal1, no candidates surpassed the threshold, meaning no rules were generated.

For Goal2, four candidates surpassed the correlation threshold, with values between 0.36 and 0.72. This resulted in the construction of six expanded rules, with correlation values ranging from 0.48 to 1. All six rules were final rules. In two cases, the maximum correlation of 1 was reached. For two other rules, the minimum frequency threshold of 0.01 was reached, preventing further expansion. Two additional rules contained candidates with the maximum index, which also prevented further expansion.

\section{Conclusion}

An efficient algorithm for constructing association rules with increasing confidence has been proposed for databases that include a target parameter.

The algorithm enables the generation of both frequent and rare rules, including those with low occurrence frequency but high confidence.

Additionally, the algorithm supports the construction of negative rules, where the presence of a rule premise negates the possibility of the target occurring.

Performance evaluations of the algorithm have been provided, supported by experimental results.

Furthermore, a theoretical justification for the proposed method of association rule construction has been presented.

\end{document}